\documentclass[a4paper,UKenglish,cleveref,autoref,thm-restate]{lipics-v2021}

\pdfoutput=1
\hideLIPIcs

\usepackage{xcolor}

\definecolor{keywordcolor}{rgb}{0.7, 0.1, 0.1}
\definecolor{tacticcolor}{rgb}{0.0, 0.1, 0.6}
\definecolor{commentcolor}{rgb}{0.4, 0.4, 0.4}
\definecolor{symbolcolor}{rgb}{0.0, 0.1, 0.6}
\definecolor{sortcolor}{rgb}{0.1, 0.5, 0.1}
\definecolor{attributecolor}{rgb}{0.7, 0.1, 0.1}

\title{An Elementary Formal Proof of the Group Law on Weierstrass Elliptic Curves in any Characteristic}

\titlerunning{The Group Law on Weierstrass Elliptic Curves}

\author{David Kurniadi Angdinata}{London School of Geometry and Number Theory, United Kingdom \and \url{https://multramate.github.io}}{ucahdka@ucl.ac.uk}{https://orcid.org/0000-0003-2096-1864}{}

\author{Junyan Xu}{Cancer Data Science Laboratory, National Cancer Institute, United States}{junyanxu.math@gmail.com}{https://orcid.org/0000-0002-3789-2319}{}

\authorrunning{D K Angdinata and J Xu}

\Copyright{David Kurniadi Angdinata and Junyan Xu}

\ccsdesc[100]{Theory of computation~Interactive proof systems}
\ccsdesc[100]{Security and privacy~Logic and verification}
\ccsdesc[100]{Mathematics of computing~Mathematical software}

\keywords{formal math, algebraic geometry, elliptic curve, group law, Lean, mathlib}

\category{}

\relatedversion{}

\supplement{The source code can be found in the \texttt{algebraic\_geometry/elliptic\_curve} folder in the \texttt{mathlib} fork \url{https://github.com/alreadydone/mathlib/tree/associativity}}

\acknowledgements{We thank the Lean community for their continual support. We thank the \texttt{mathlib} contributors, especially Anne Baanen, for developing libraries this work depends on. We thank Marc Masdeu and Michael Stoll for proposing alternative proofs. DKA would like to thank Kevin Buzzard for his guidance and Mel Levin for suggesting the formalisation in the first place. \\
This work was supported by the Engineering and Physical Sciences Research Council [EP/S021590/1], EPSRC Centre for Doctoral Training in Geometry and Number Theory (London School of Geometry and Number Theory), University College London. This research was supported in part by the Intramural Research Program of the Center for Cancer Research, National Cancer Institute, NIH.}

\nolinenumbers

\EventEditors{Editors}
\EventNoEds{0}
\EventLongTitle{14th International Conference on Interactive Theorem Proving (ITP 2023)}
\EventShortTitle{ITP 2023}
\EventAcronym{ITP}
\EventYear{2023}
\EventDate{Date}
\EventLocation{Location}
\EventLogo{}
\SeriesVolume{Volume}
\ArticleNo{No}

\begin{document}

\maketitle

\begin{abstract}
Elliptic curves are fundamental objects in number theory and algebraic geometry, whose points over a field form an abelian group under a geometric addition law. Any elliptic curve over a field admits a Weierstrass model, but prior formal proofs that the addition law is associative in this model involve either advanced algebraic geometry or tedious computation, especially in characteristic two. We formalise in the Lean theorem prover, the type of nonsingular points of a Weierstrass curve over a field of any characteristic and a purely algebraic proof that it forms an abelian group.
\end{abstract}

\section{Introduction}

\subsection{Elliptic curves}

In its earliest form, algebraic geometry is the branch of mathematics studying the solutions to systems of polynomial equations over a base field $ F $, namely sets of the form
\[ \{(x_1, \dots, x_n) \in F^n : f_1(x_1, \dots, x_n) = 0, \ \dots, \ f_k(x_1, \dots, x_n) = 0\}, \]
for some polynomials $ f_i \in F[X_1, \dots, X_n] $. These are called \textbf{affine varieties}, and they can be endowed with topologies and a notion of morphisms which makes them simultaneously geometric objects. General \textbf{varieties} are locally modelled on affine ones, with morphisms between them locally given by polynomials, and are often classified by their geometric properties such as \emph{smoothness}, or invariants such as the \emph{dimension} or the \emph{genus}.

Having dimension one and genus one, \emph{elliptic curves} are amongst the simplest varieties with respect to these geometric notions, and their set of points can be endowed with the structure of an abelian group. When the base field is the rationals $ \mathbb{Q} $, a common definition uses the \emph{short Weierstrass model}, given by the equation $ y^2 = x^3 + ax + b $ for some fixed $ a, b \in \mathbb{Q} $, and its group law can be defined explicitly by quotients of polynomial functions.

\pagebreak

Elliptic curves are blessed with an extremely rich theory, spanning the fields of algebraic geometry, complex analysis, number theory, representation theory, dynamical systems, and even information security. The \emph{Birch and Swinnerton-Dyer conjecture} \cite{bsd} in number theory, one of the seven \emph{Millennium Prize Problems}, is an equality between an analytic quantity of an elliptic curve over $ \mathbb{Q} $ and an algebraic quantity defined in terms of its group structure. Their close relation with modular forms is precisely the content of the \emph{Taniyama--Shimura conjecture} proven by Andrew Wiles \cite{wiles}, which implies \emph{Fermat's last theorem} and laid the foundations of the \emph{Langlands programme}, a web of influential conjectures linking number theory and geometry described to be ``kind of a grand unified theory of mathematics'' \cite{langlands}. Outside of mathematics, elliptic curves over finite fields see applications in primality proving \cite{atkin} and integer factorisation \cite{lenstra}, as well as in public key cryptography \cite{ecdh}.


\subsection{Formalisation attempts}

The group law on an elliptic curve in the short Weierstrass model over a field $ F $ has been formalised in several theorem provers, \footnote{see \cref{sec:relate} for related work} but this model fails to be an elliptic curve when $ \mathrm{char}(F) = 2 $, and there has been no known successful attempts to formalise the group law in a \emph{universal} model that captures all elliptic curves in all $ \mathrm{char}(F) $. The issue was that a computational proof of associativity in any universal model is, as Russinoff described, ``an elementary but computationally intensive arithmetic exercise'' involving massive polynomials \cite{sutherland}, \footnote{see \cref{sec:experiment} for experimental results} while a typical conceptual proof is ``a deep theorem of algebraic or projective geometry'' requiring prior formalisation of advanced geometric constructions, with ``evidence that an elementary hand proof of this result is a practical impossibility'' \cite{russinoff}. On the other hand, having the group law in $ \mathrm{char}(F) = 2 $ is necessary for certain applications, such as the proof of Fermat's last theorem. It is less crucial from an information security viewpoint, as curves over binary fields are prone to attacks and no longer recommended by NIST \cite{nist}.

We give a conceptual yet computation-light proof of the group law in the full \emph{Weierstrass model}, which is universal in all $ \mathrm{char}(F) $. The argument is purely algebraic and easily surveyable, in the sense that all logical deductions and necessary computations can be performed by hand in a matter of minutes. The proof is formalised in \emph{Lean 3} \cite{lean}, an interactive theorem prover based on the calculus of constructions with inductive types, and is integrated as part of its monolithic mathematical library \texttt{mathlib} \cite{mathlib}. The implementation extensively uses existing constructions in the linear algebra and ring theory libraries of \texttt{mathlib}, particularly constructions and results surrounding \texttt{algebra.norm} and \texttt{class\_group} \cite{docs}. The relevant code is primarily split into two files \texttt{weierstrass.lean} and \texttt{point.lean} under the folder \texttt{algebraic\_geometry/elliptic\_curve} in the \texttt{associativity} branch of \texttt{mathlib}, both of which are currently undergoing reviews for their merge. With this paper, we hope that our simple proof will be replicated and will open the way for the formalisation of elliptic curve cryptography in many other theorem provers, which has been a major motivation of recent formalisation attempts \cite{russinoff, fox, hales, bartzia}.

The remainder of this paper will describe the relevant constructions (\cref{sec:weierstrass}), detail the mathematical argument of the proof (\cref{sec:group}), and discuss implementation considerations (\cref{sec:discuss}). Throughout, definitions will be described in code snippets where relevant, but proofs of lemmas will be outlined mathematically for the sake of clarity. The reader may refer to the \texttt{mathlib} documentation \cite{docs} for definitions and lemmas involved.

\pagebreak

\section{Weierstrass equations}
\label{sec:weierstrass}

Let $ F $ be a field. In the sense of modern algebraic geometry, an \textbf{elliptic curve} $ E $ over $ F $ is a smooth projective curve \footnote{variety of dimension one} of genus one equipped with a base point $ \mathcal{O}_E \in E $ with coordinates in $ F $. More concretely, any elliptic curve over $ F $ admits a model in the projective plane $ \mathbb{P}_F^2 $ defined by an explicit polynomial equation in homogeneous coordinates $ [X : Y : Z] $.

\begin{proposition}
\label{prop:weierstrass}
If $ E $ is an elliptic curve over $ F $, then there are rational functions $ x, y : E \to F $ such that the map $ \phi : E \to \mathbb{P}_F^2 $ given by $ \phi(P) = [x(P) : y(P) : 1] $ maps $ \mathcal{O}_E $ to $ [0 : 1 : 0] $, and defines an isomorphism of varieties between $ E $ and the curve
\[ Y^2Z + a_1XYZ + a_3YZ^2 = X^3 + a_2X^2Z + a_4XZ^2 + a_6Z^3, \]
for some coefficients $ a_i \in F $. Conversely, any curve $ W $ in $ \mathbb{P}_F^2 $ given by such an equation, with coefficients $ a_i \in F $, is an elliptic curve over $ F $ with base point $ [0 : 1 : 0] $ if the quantity
\begin{align*}
\Delta_W := & -(a_1^2 + 4a_2)^2(a_1^2a_6 + 4a_2a_6 - a_1a_3a_4 + a_2a_3^2 - a_4^2) - 8(2a_4 + a_1a_3)^3 \\
& - 27(a_3^2 + 4a_6)^2 + 9(a_1^2 + 4a_2)(2a_4 + a_1a_3)(a_3^2 + 4a_6) \in F
\end{align*}
is nonzero.
\end{proposition}

\begin{proof}
This is a consequence of the Riemann--Roch theorem \cite[Proposition III.3.1]{silverman}.
\end{proof}

This is the \textbf{Weierstrass model} of $ E $, and such a curve is called a \textbf{Weierstrass curve}, whose corresponding \textbf{Weierstrass equation} in the affine chart $ Z \ne 0 $ is
\[ Y^2 + a_1XY + a_3Y = X^3 + a_2X^2 + a_4X + a_6. \]
In this model, the smoothness condition on $ E $ becomes equivalent to the \textbf{discriminant} $ \Delta_E \in F $ being nonzero, so an elliptic curve over $ F $ may instead be defined as a Weierstrass curve with the discriminant condition, which is more amenable for computational purposes.

\subsection{Weierstrass curves}

Let $ F $ be a commutative ring, and let $ W $ be a Weierstrass curve over $ F $. Explicitly, this is merely the data of five coefficients $ a_1, a_2, a_3, a_4, a_6 \in F $, noting that the Weierstrass equation is not visible at this stage. For the sake of generality, the structure \texttt{weierstrass\_curve} is defined more generally over an arbitrary type $ F $, but all subsequent constructions will assume that $ F $ is at least a commutative ring. The structure \texttt{elliptic\_curve} then \texttt{extends weierstrass\_curve} by adding the data of an element $ \Delta' $ in the group of units $ F^\times $ of $ F $ and a proof that $ \Delta' = \Delta_W $, so most definitions for \texttt{weierstrass\_curve} carry over automatically.
\begin{lstlisting}
structure weierstrass_curve (F : Type) := (a₁ a₂ a₃ a₄ a₆ : F)

structure elliptic_curve (F : Type) [comm_ring F] extends weierstrass_curve F :=
  (Δ' : Fˣ) (coe_Δ' : ↑Δ' = to_weierstrass_curve.Δ)
\end{lstlisting}
Here, \texttt{to\_weierstrass\_curve} is a function generated automatically by the \texttt{extends} keyword, which projects an \texttt{elliptic\_curve} down to its underlying \texttt{weierstrass\_curve} by forgetting $ \Delta' $ and the proof that $ \Delta' = \Delta_W $. Note that \texttt{elliptic\_curve} was originally defined in one stretch by Buzzard, but is now refactored through the more general \texttt{weierstrass\_curve}.

\pagebreak

\begin{remark}
\label{rem}
This definition of an elliptic curve is universal over a large class of commutative rings, namely those with trivial Picard group \cite[Section 2.2]{katz}, which includes fields, and also local rings and unique factorisation domains. In general, an elliptic curve $ E $ may be defined relative to an arbitrary scheme $ S $ as a smooth proper morphism $ E \to S $ in the category of schemes whose geometric fibres are all integral curves of genus one, equipped with a section $ S \to E $ that plays the role of the base point $ \mathcal{O}_E $ for all fibres simultaneously. When $ S $ is the spectrum of such a commutative ring, the Riemann--Roch theorem can be generalised so that $ E $ remains isomorphic to a Weierstrass curve, but over a general commutative ring, $ E $ only has Weierstrass equations locally that may not patch together to form a global equation.
\end{remark}

The discriminant $ \Delta_W \in F $ is expressed entirely in terms of the five coefficients, but it is clearer to extract intermediate quantities \cite[Section III.1]{silverman} to simplify the large expression.
\begin{lstlisting}
namespace weierstrass_curve

variables {F : Type} [comm_ring F] (W : weierstrass_curve F)

def b₂ : F := W.a₁^2 + 4*W.a₂
def b₄ : F := 2*W.a₄ + W.a₁*W.a₃
def b₆ : F := W.a₃^2 + 4*W.a₆
def b₈ : F := W.a₁^2*W.a₆ + 4*W.a₂*W.a₆ - W.a₁*W.a₃*W.a₄ + W.a₂*W.a₃^2 - W.a₄^2

def Δ : F := -W.b₂^2*W.b₈ - 8*W.b₄^3 - 27*W.b₆^2 + 9*W.b₂*W.b₄*W.b₆
\end{lstlisting}
Here, dot notation allows for the fields corresponding to the five coefficients $ a_i \in F $ to be accessed as \texttt{W.a\ensuremath{_i}}, and living inside the \texttt{namespace weierstrass\_curve} means that the quantities $ b_i \in F $ and $ \Delta \in F $ may also be accessed as \texttt{W.b\ensuremath{_i}} and \texttt{W.\ensuremath{\Delta}} respectively.

These quantities are indexed as such as a result of their transformation upon applying the linear change of variables $ (X, Y) \mapsto (u^2X + r, u^3Y + u^2sX + t) $, for some $ u \in F^\times $ and some $ r, s, t \in F $. In fact, these are all the possible isomorphisms of varieties between elliptic curves in the Weierstrass model \cite[Proposition III.3.1]{silverman}.
\begin{lstlisting}
variables (u : Fˣ) (r s t : F)

@[simps] def variable_change : weierstrass_curve F :=
  { a₁ := ↑u⁻¹*(W.a₁ + 2*s),
    a₂ := ↑u⁻¹^2*(W.a₂ - s*W.a₁ + 3*r - s^2),
    a₃ := ↑u⁻¹^3*(W.a₃ + r*W.a₁ + 2*t),
    a₄ := ↑u⁻¹^4*(W.a₄ - s*W.a₃ + 2*r*W.a₂ - (t + r*s)*W.a₁ + 3*r^2 - 2*s*t),
    a₆ := ↑u⁻¹^6*(W.a₆ + r*W.a₄ + r^2*W.a₂ + r^3 - t*W.a₃ - t^2 - r*t*W.a₁) }

@[simp] lemma variable_change_b₂ : (W.variable_change u r s t).b₂ = ↑u⁻¹^2*(...)
@[simp] lemma variable_change_b₄ : (W.variable_change u r s t).b₄ = ↑u⁻¹^4*(...)
@[simp] lemma variable_change_b₆ : (W.variable_change u r s t).b₆ = ↑u⁻¹^6*(...)
@[simp] lemma variable_change_b₈ : (W.variable_change u r s t).b₈ = ↑u⁻¹^8*(...)

@[simp] lemma variable_change_Δ : (W.variable_change u r s t).Δ = ↑u⁻¹^12*W.Δ
\end{lstlisting}
Here, \texttt{variable\_change} defines a new \texttt{weierstrass\_curve} under the change of variables by explicitly stating how each of the five coefficients are transformed, and is tagged with \texttt{simps} to automatically generate \texttt{simp} lemmas corresponding to each of the five projections. The exact expressions for the transformed quantities are not too important, but note that their indices precisely correspond to the exponent of $ u^{-1} \in F^\times $ in the transformation.

\pagebreak

\subsection{Equations and nonsingularity}

Now consider the polynomial in $ F[X, Y] $ associated to $ W $ denoted by
\[ W(X, Y) := Y^2 + (a_1X + a_3)Y - (X^3 + a_2X^2 + a_4X + a_6), \]
so that the Weierstrass equation literally reads $ W(X, Y) = 0 $, and its partial derivatives
\[ W_X(X, Y) := a_1Y - (3X^2 + 2a_2X + a_4), \qquad W_Y(X, Y) := 2Y + a_1X + a_3. \]
When they do not simultaneously vanish when evaluated at an affine point $ (x, y) \in W $, the affine point is said to be \textbf{nonsingular}. This is implemented slightly confusingly as follows. \footnote{this representation of bivariate polynomials will be explained in \cref{sec:implementation}}
\begin{lstlisting}
def polynomial : F[X][Y] :=
  Y^2 + C (C W.a₁*X + C W.a₃)*Y - C (X^3 + C W.a₂*X^2 + C W.a₄*X + C W.a₆)
def equation (x y : F) : Prop := (W.polynomial.eval (C y)).eval x = 0

def polynomial_X : F[X][Y] := C (C W.a₁)*Y - C (C 3*X^2 + C (2*W.a₂)*X + C W.a₄)
def polynomial_Y : F[X][Y] := C (C 2)*Y + C (C W.a₁*X + C W.a₃)
def nonsingular (x y : F) : Prop :=
  W.equation x y ∧ ((W.polynomial_X.eval (C y)).eval x ≠ 0
                    ∨ (W.polynomial_Y.eval (C y)).eval x ≠ 0)
\end{lstlisting}
Here, \texttt{F[X][Y]} denotes the polynomial ring over the polynomial ring over \texttt{F}, to simulate the bivariate polynomial ring $ F[X, Y] $ over $ F $. The outer variable $ Y $ is denoted by the symbol \texttt{Y} and the inner variable $ X $ is denoted by the constant function \texttt{C} applied to the symbol \texttt{X}, while actual constants are enclosed in two layers of the constant function \texttt{C}.

\begin{remark}
This definition of nonsingularity in terms of partial derivatives is valid and convenient when the base ring is a field, but over a general commutative ring the same notion should be characterised locally with a notion of smoothness relative to the base spectrum.
\end{remark}

Many properties of Weierstrass curves remain invariant under the aforementioned changes of variables, and it is often easier to prove results for Weierstrass equations with fewer terms. For instance, if $ \mathrm{char}(F) \ne 2 $, then $ (X, Y) \mapsto (X, Y - \tfrac{1}{2}a_1X - \tfrac{1}{2}a_3) $ completes the square for $ W(X, Y) $ and eliminates the $ XY $ and $ Y $ terms, and if further $ \mathrm{char}(F) \ne 3 $, then $ (X, Y) \mapsto (X - \tfrac{1}{12}b_2, Y) $ completes the cube for $ W(X, Y) $ and eliminates the $ X^2 $ term as well.

Perhaps a more prominent application is the proof that, for any affine point $ (x, y) \in W $, the nonvanishing of $ \Delta_W $ implies that $ (x, y) $ is nonsingular. This statement is easy for $ (x, y) = (0, 0) $, since $ (0, 0) \in W $ implies that $ a_6 = 0 $, and $ (0, 0) $ being singular implies that $ a_3 = a_4 = 0 $, so that $ \Delta_W = 0 $ by a simple substitution. On the other hand, for any $ (x, y) \in F^2 $, the change of variables $ (X, Y) \mapsto (X - x, Y - y) $ merely translates $ W $ so that $ (0, 0) $ gets mapped to $ (x, y) $, so the same statement clearly also holds for $ (x, y) $.
\begin{lstlisting}
lemma nonsingular_zero : W.nonsingular 0 0 ↔ W.a₆ = 0 ∧ (W.a₃ ≠ 0 ∨ W.a₄ ≠ 0)
lemma nonsingular_zero_of_Δ_ne_zero (h : W.equation 0 0) (hΔ : W.Δ ≠ 0) :
  W.nonsingular 0 0
lemma nonsingular_iff_variable_change (x y : F) :
  W.nonsingular x y ↔ (W.variable_change 1 x 0 y).nonsingular 0 0
lemma nonsingular_of_Δ_ne_zero {x y : F} (h : W.equation x y) (hΔ : W.Δ ≠ 0) :
  W.nonsingular x y
\end{lstlisting}

In fact, it is also true that $ \Delta_W \ne 0 $ if every point over the algebraic closure is nonsingular \cite[Proposition III.1.4]{silverman}, but the proof is more difficult and has yet to be formalised.

\pagebreak

\subsection{Point addition}

The set of points on an elliptic curve can be endowed with an \textbf{addition law} defined by a geometric secant-and-tangent process enabled by Vieta's formulae. \footnote{if a cubic polynomial has two roots in a field, then its third root is also in the field} This can be easily described in the Weierstrass model, where a point on $ W $ is either of the form $ (x, y) $ in the affine chart $ Z \ne 0 $ and satisfies the Weierstrass equation, or is the unique point at infinity $ \mathcal{O}_W := [0 : 1 : 0] $ when $ Z = 0 $. If $ S \in W $ is a singular point, the same geometric process will yield $ P + S = S = S + P $ for any other point $ P \in W $, so it necessitates considering only nonsingular points on $ W $ to obtain a group \cite[Section III.2]{silverman}. Note that if $ W $ is an elliptic curve, then all points are nonsingular by \texttt{nonsingular\_of\_\ensuremath{\Delta}\_ne\_zero}.
\begin{lstlisting}
inductive point
  | zero
  | some {x y : F} (h : W.nonsingular x y)

namespace point
\end{lstlisting}
Here, \texttt{zero} refers to $ \mathcal{O}_W $ and \texttt{some} refers to an affine point on $ W $. Note that a proof that an affine point $ (x, y) \in W $ is nonsingular already subsumes the data of the coordinates $ (x, y) \in F^2 $ in its type, so such a point is constructed by giving only the proof.

\begin{remark}
The set of points defined here will later be shown to form an abelian group under this addition law, but the presence of division means that the base ring needs to be a field. Over a general commutative ring $ R $ these should be replaced with scheme-theoretic points $ \mathrm{Spec}(R) \to E $, and the elliptic curve acquires the structure of a group scheme.
\end{remark}

In this model, the identity element in the group of points is defined to be $ \mathcal{O}_W \in W $.
\begin{lstlisting}
instance : has_zero W.point := ⟨zero⟩
\end{lstlisting}
Here, the \texttt{instance} of \texttt{has\_zero} allows the notation \texttt{0} instead of \texttt{zero}.

Given a nonsingular point $ P \in W $, its negation $ -P $ is defined to be the unique third point of intersection between $ W $ and the line through $ \mathcal{O}_W $ and $ P $, which is vertical when drawn on the affine plane. Explicitly, if $ P := (x, y) $, then $ -P := (x, \sigma_x(y)) $, where
\[ \sigma_X(Y) := -Y - (a_1X + a_3) \]
is the \textbf{negation polynomial}, which gives a very useful involution of the affine plane.
\begin{lstlisting}
def neg_polynomial : F[X][Y] := -Y - C (C W.a₁*X + C W.a₃)
def neg_Y (x y : F) : F := (W.neg_polynomial.eval (C y)).eval x
\end{lstlisting}
Here, \texttt{neg\_Y} is defined in terms of \texttt{neg\_polynomial} for clarity, but its actual definition is written out as \texttt{-y - C (C W.a\ensuremath{_1}*x + C W.a\ensuremath{_3})}. This is merely to avoid requiring the \texttt{noncomputable} tag, since polynomial operations are currently \texttt{noncomputable} in \texttt{mathlib}.

To define negation, it remains to prove that the negation of a nonsingular point on $ W $ is again a nonsingular point on $ W $, but this is straightforward.

\begin{lemma}
\label{lem:nonsingular_neg}
If $ (x, y) \in W $ is nonsingular, then $ (x, \sigma_x(y)) \in W $ is nonsingular.
\end{lemma}

\begin{proof}
Since $ (x, y) \in W $, verifying that $ W(x, y) = W(x, \sigma_x(y)) $ gives $ (x, \sigma_x(y)) \in W $ as well. Now assume that $ W_Y(x, \sigma_x(y)) = 0 $. It can be checked that $ y = \sigma_x(y) $, so $ W_Y(x, y) = 0 $ as well. Since $ (x, y) $ is nonsingular, $ W_X(x, y) \ne 0 $, so $ W_X(x, \sigma_x(y)) \ne 0 $ as well.
\end{proof}

\pagebreak

\cref{lem:nonsingular_neg} is \texttt{nonsingular\_neg}, which maps a proof that $ (x, y) \in W $ is nonsingular to a proof that $ -(x, y) \in W $ is nonsingular. This leads to the definition of negation.
\begin{lstlisting}
def neg : W.point → W.point
  | 0        := 0
  | (some h) := some (nonsingular_neg h)

instance : has_neg W.point := ⟨neg⟩
\end{lstlisting}
Here, the \texttt{instance} of \texttt{has\_neg} allows the notation \texttt{-P} instead of \texttt{neg P}.

Given two nonsingular points $ P_1, P_2 \in W $, their sum $ P_1 + P_2 $ is defined to be the negation of the unique third point of intersection between $ W $ and the line through $ P_1 $ and $ P_2 $, which again exists by B\'ezout's theorem. Explicitly, let $ P_1 := (x_1, y_1) $ and $ P_2 := (x_2, y_2) $.
\begin{itemize}
\item If $ x_1 = x_2 $ and $ y_1 = \sigma_{x_2}(y_2) $, then this line is vertical and $ P_1 + P_2 := \mathcal{O}_W $.
\item If $ x_1 = x_2 $ and $ y_1 \ne \sigma_{x_2}(y_2) $, then this line is the tangent of $ W $ at $ P_1 = P_2 $, and has slope
\[ \ell := \dfrac{-W_X(x_1, y_1)}{W_Y(x_1, y_1)}. \]
\item Otherwise $ x_1 \ne x_2 $, then this line is the secant of $ W $ through $ P_1 $ and $ P_2 $, and has slope
\[ \ell := \dfrac{y_1 - y_2}{x_1 - x_2}. \]
\end{itemize}
In the latter two cases, the \textbf{line polynomial}
\[ \lambda(X) = \lambda_{x_1, y_1, \ell}(X) := \ell(X - x_1) + y_1. \]
can be shown to satisfy $ \lambda(x_1) = y_1 $ and $ \lambda(x_2) = y_2 $, so that $ x_1 $ and $ x_2 $ are two roots of the \textbf{addition polynomial} $ W(X, \lambda(X)) $, obtained by evaluating $ W(X, Y) $ at $ \lambda(X) $, where $ W(X, Y) $ is viewed as a polynomial over $ F[X] $. In an attempt to reduce code duplication for the different cases, these accept an additional parameter \texttt{L} for the slope $ \ell $.
\begin{lstlisting}
def line_polynomial (x y L : F) : F[X] := C L * (X - C x) + C y
def add_polynomial (x y L : F) : F[X] := W.polynomial.eval (line_polynomial x y L)
\end{lstlisting}

The $ X $-coordinate of $ P_1 + P_2 $ is then the third root of $ W(X, \lambda(X)) $, so that
\begin{equation}
\label{eq:1}
W(X, \lambda(X)) = -(X - x_1)(X - x_2)(X - x_3).
\end{equation}
By inspecting the $ X^2 $ terms of $ W(X, \lambda(X)) $, this third root can be shown to be
\[ x_3 := \ell^2 + a_1\ell - a_2 - x_1 - x_2, \]
so the $ Y $-coordinate of $ -(P_1 + P_2) $ is
\[ y_3' := \lambda(x_3), \]
and that of $ P_1 + P_2 $ is
\[ y_3 := \sigma_{x_3}(y_3'). \]
These correspond to the coordinate functions \texttt{add\_X}, \texttt{add\_Y'}, and \texttt{add\_Y} respectively.
\begin{lstlisting}
def add_X (x₁ x₂ L : F) : F := L^2 + W.a₁*L - W.a₂ - x₁ - x₂
def add_Y' (x₁ x₂ y₁ L : F) : F := (line_polynomial x₁ y₁ L).eval (W.add_X x₁ x₂ L)
def add_Y (x₁ x₂ y₁ L : F) : F := W.neg_Y (W.add_X x₁ x₂ L) (W.add_Y' x₁ x₂ y₁ L)
\end{lstlisting}
Here, \texttt{add\_Y'} is defined in terms of \texttt{line\_polynomial} and \texttt{add\_X}, but in actuality it is again written out in the evaluated form \texttt{C L*(X - C x\ensuremath{_1}) + C y\ensuremath{_1}} to avoid the \texttt{noncomputable} tag.

\pagebreak

The slope itself is defined as a conditional over the three cases, and since two of them involve division, this is the first occasion where $ W $ needs to be defined over a field $ F $.
\begin{lstlisting}
variables {F : Type} [field F] (W : weierstrass_curve F)

def slope (x₁ x₂ y₁ y₂ : F) : F :=
  if hx : x₁ = x₂ then
    if hy : y₁ = W.neg_Y x₂ y₂ then 0
    else (3*x₁^2 + 2*W.a₂*x₁ + W.a₄ - W.a₁*y₁)/(y₁ - W.neg_Y x₁ y₁)
  else (y₁ - y₂)/(x₁ - x₂)
\end{lstlisting}
Note that \texttt{slope} returns the \emph{junk value} of $ 0 \in F $ when the slope is vertical. This practice of assigning a junk value is common in \texttt{mathlib} to avoid excessive layers of \texttt{option}, and any useful result proven using such a definition would include a condition so that this junk value will never be reached. In the case of \texttt{slope}, this is the implication $ x_1 = x_2 \to y_1 \ne \sigma_{x_2}(y_2) $, which holds precisely either when $ x_1 \ne x_2 $, or when $ x_1 = x_2 $ but $ (x_1, y_1) \ne -(x_2, y_2) $.
\begin{lstlisting}
variables {x₁ x₂ y₁ y₂ : F} (hxy : x₁ = x₂ → y₁ ≠ W.neg_Y x₂ y₂)

example (hx : x₁ ≠ x₂) : x₁ = x₂ → y₁ ≠ W.neg_Y x₂ y₂ := λ h, (hx h).elim
example (hy : y₁ ≠ W.neg_Y x₂ y₂) : x₁ = x₂ → y₁ ≠ W.neg_Y x₂ y₂ := λ _, hy
\end{lstlisting}
Here, the examples return proofs of \texttt{hxy} assuming $ x_1 \ne x_2 $ and $ y_1 \ne \sigma_{x_2}(y_2) $ respectively. They are illustrated here for clarity, but they do not exist in the actual Lean code since their term mode proofs are short enough to be inserted directly whenever necessary.

To define addition, it remains to prove that the addition of two nonsingular points on $ W $ is again a nonsingular point on $ W $. This is slightly lengthy but purely conceptual.

\begin{lemma}
\label{lem:nonsingular_add}
If $ (x_1, y_1), (x_2, y_2) \in W $ are nonsingular, then $ (x_3, y_3) \in W $ is nonsingular.
\end{lemma}

\begin{proof}
By \texttt{nonsingular\_neg}, it suffices to prove that $ (x_3, \lambda(x_3)) = (x_3, y_3') $ is nonsingular, since $ (x_3, \lambda(x_3)) \in W $ is clear. Taking derivatives of both sides in (\ref{eq:1}) yields
\[ W_X(X, \lambda(X)) + \ell \cdot W_Y(X, \lambda(X)) = -((X - x_1)(X - x_2) + (X - x_1)(X - x_3) + (X - x_2)(X - x_3)), \]
which does not vanish at $ X = x_3 $, so that $ W(X, \lambda(X)) $ has at least one nonvanishing partial derivative, unless possibly when $ x_3 \in \{x_1, x_2\} $. The latter implies that $ (x_3, \lambda(x_3)) \in \{\pm(x_1, y_1), \pm(x_2, y_2)\} $, but these are nonsingular by assumption or by \texttt{nonsingular\_neg}.
\end{proof}

\cref{lem:nonsingular_add} is \texttt{nonsingular\_add}, which accepts a proof that $ (x_1, y_1) \in W $ is nonsingular, a proof that $ (x_2, y_2) \in W $ is nonsingular, and a proof of \texttt{hxy}, and returns a proof that $ (x_1, y_1) + (x_2, y_2) \in W $ is nonsingular. This finally leads to the definition of addition.
\begin{lstlisting}
def add : W.point → W.point → W.point
  | 0                      P                     := P
  | P                      0                     := P
  | (@some _ _ _ x₁ y₁ h₁) (@some _ _ _ x₂ y₂ h₂) :=
    if hx : x₁ = x₂ then
      if hy : y₁ = W.neg_Y x₂ y₂ then 0
      else some (nonsingular_add h₁ h₂ (λ _, hy))
    else some (nonsingular_add h₁ h₂ (λ h, (hx h).elim))

instance : has_add W.point := ⟨add⟩
\end{lstlisting}
Here, the \texttt{instance} of \texttt{has\_add} allows the notation \texttt{P\ensuremath{_1} + P\ensuremath{_2}} instead of \texttt{add P\ensuremath{_1} P\ensuremath{_2}}. The annotation \texttt{@} for \texttt{some} simply gives access to all implicit variables in its definition, particularly $ x_1, x_2, y_1, y_2 \in F $ that is necessary to even state the conditions \texttt{hx} and \texttt{hy}.

\pagebreak

\section{Group law}
\label{sec:group}

Let $ W $ be a Weierstrass curve over a field $ F $, and denote its set of nonsingular points by $ W(F) $. The addition law defined in the previous section is in fact a \textbf{group law}.

\begin{proposition}
\label{prop:group}
$ W(F) $ forms an abelian group under this addition law.
\end{proposition}

The axioms of this group law are mostly straightforward, typically just by examining the definition for each of the five cases. For instance, the \texttt{lemma add\_left\_neg} that says $ -P + P = \mathcal{O}_W $ is immediate, since $ -\mathcal{O}_W + \mathcal{O}_W = \mathcal{O}_W $ by definition, and $ -(x, y) + (x, y) = (x, \sigma_x(y)) + (x, y) = \mathcal{O}_W $ for any $ (x, y) \in W(F) $ by the first case of affine addition.

On the other hand, associativity is far from being straightforward. \footnote{see \cref{sec:experiment} for alternative proofs} A notable algebro-geometric proof involves canonically identifying $ W(F) $ with its \emph{Picard group}, a natural geometric construction associated to $ W $ with a known group structure, and proving that this identification respects the addition law on $ W $ \cite[Proposition III.3.4]{silverman}. This is generally regarded as the most conceptual proof, as it explains the seemingly arbitrary secant-and-tangent process, and more crucially because it works for any $ \mathrm{char}(F) $.

The proof in this paper is an analogue of this proof, but the arguments involved are purely algebraic without the need for any geometric machinery, in contrast to the typical algebro-geometric proof. The main idea, originally inspired by Buzzard's post on Zulip \cite{buzzard} and Chapman's answer on MathOverflow \cite{chapman}, is to construct an explicit function \texttt{to\_class} from $ W(F) $ to the \emph{ideal class group} $ \mathrm{Cl}(R) $ of an integral domain $ R $ associated to $ W $, then to prove that this function is injective and respects the addition law on $ W $. This is a construction in commutative algebra whose definition will now be briefly outlined, but for a more comprehensive introduction to ideal class groups motivated by arithmetic examples, and especially specific details of their formalisation in \texttt{mathlib}, the reader is strongly urged to consult the original paper by Baanen, Dahmen, Narayanan, and Nuccio \cite[Section 2]{baanen}.

\subsection{Ideal class group of the coordinate ring}

Given an integral domain $ R $ with a fraction field $ K $, a \textbf{fractional ideal} of $ R $ is simply a $ R $-submodule $ I $ of $ K $ such that $ x \cdot I \subseteq R $ for some nonzero $ x \in R $. This generalises the notion of an ideal of $ R $, since any ideal is clearly a fractional ideal with $ x = 1 $, so ideals are sometimes referred to as \textbf{integral ideals} to distinguish them from fractional ideals.

In \texttt{mathlib}, this is expressed as a transitive coercion from \texttt{ideal} to \texttt{fractional\_ideal}.
\begin{lstlisting}
instance : has_coe_t (ideal R) (fractional_ideal R⁰ (fraction_ring R))
\end{lstlisting}
Here, $ R^0 $ is the submonoid of non-zero-divisors of $ R $, and \texttt{fraction\_ring} returns the canonical choice of a fraction field of $ R $ obtained by adjoining inverses of elements of $ R^0 $. Since $ R $ is an integral domain in this case, all nonzero elements of $ R $ become invertible in its fraction field.

Analogously to integral ideals, the set of fractional ideals of $ R $ forms a commutative semiring under the usual operations of addition and multiplication for submodules. A fractional ideal $ I $ of $ R $ is \textbf{invertible} if $ I \cdot J = R $ for some fractional ideal $ J $ of $ R $, and the subset of invertible fractional ideals of $ R $ forms an abelian group under multiplication. An important class of invertible fractional ideals are those generated by a nonzero element of $ K $, called \textbf{principal fractional ideals}. The \textbf{ideal class group} $ \mathrm{Cl}(R) $ of $ R $ is then defined to be the quotient group of invertible fractional ideals by principal fractional ideals.

\pagebreak

In \texttt{mathlib}, a \texttt{class\_group} element is constructed from an invertible \texttt{fractional\_ideal} via \texttt{class\_group.mk}, and this association is a \texttt{monoid\_hom} that respects multiplication.
\begin{lstlisting}
def class_group.mk : (fractional_ideal R⁰ K)ˣ →* class_group R := ...
\end{lstlisting}
Here, it is worth noting that $ \mathrm{Cl}(R) $ is typically defined only when $ R $ is \textbf{Dedekind}, namely when every nonzero fractional ideal is invertible, and in such domains, \texttt{class\_group.mk0} constructs a \texttt{class\_group} element directly from a nonzero \texttt{fractional\_ideal}.

The integral domain in consideration is the \textbf{coordinate ring} of $ W $, that is
\[ F[W] := F[X, Y] / \langle W(X, Y) \rangle, \]
whose fraction field is the \textbf{function field} $ F(W) := \mathrm{Frac}(F[W]) $ of $ W $.
\begin{lstlisting}
@[derive comm_ring] def coordinate_ring : Type := adjoin_root W.polynomial
abbreviation function_field : Type := fraction_ring W.coordinate_ring

namespace coordinate_ring
\end{lstlisting}
Here, $ W(X, Y) $ is viewed as a quadratic monic polynomial with coefficients in $ F[X] $, so \texttt{adjoin\_root} constructs the quotient ring $ F[W] $ by adjoining its root $ Y $. The tag \texttt{derive comm\_ring} automatically generates a \texttt{instance} of \texttt{comm\_ring} present in \texttt{adjoin\_root}, \footnote{this has a type unification performance issue that will be detailed in \cref{sec:implementation}} while \texttt{abbreviation} is just a \texttt{def} that inherits every \texttt{instance} from \texttt{fraction\_ring}.

A priori, $ F[W] $ is only a commutative ring, but for $ \mathrm{Cl}(F[W]) $ to make sense it needs to be at least an integral domain, which is straightforward from the shape of $ W(X, Y) $.

\begin{lemma}
\label{lem:is_domain}
$ F[W] $ is an integral domain.
\end{lemma}

\begin{proof}
It suffices to prove that $ W(X, Y) $ is prime, but $ F[X, Y] $ is a unique factorisation domain since $ F $ is a field, so it suffices to prove that $ W(X, Y) $ is irreducible. Suppose for a contradiction that it were reducible as a product of two factors. Since it is a monic polynomial in $ Y $, the leading coefficients of the two factors multiply to $ 1 $, so without loss of generality
\[ W(X, Y) = (Y - p(X))(Y - q(X)), \]
for some polynomials $ p(X), q(X) \in F[X] $. Comparing coefficients yields
\[ a_1X + a_3 = -(p(X) + q(X)), \qquad -(X^3 + a_2X^2 + a_4X + a_6) = p(X)q(X), \]
which cannot simultaneously hold by considering $ \deg_X(p(X)) $ and $ \deg_X(q(X)) $.
\end{proof}

\cref{lem:is_domain} is formalised as an \texttt{instance} of \texttt{is\_domain} for \texttt{W.coordinate\_ring}. In fact, $ F[W] $ is also Dedekind when $ \Delta_W \ne 0 $, but this will not be necessary in the proof.

\begin{remark}
This argument with ideal class groups is essentially an algebraic translation of the algebro-geometric argument with Picard groups. An invertible fractional ideal on an integral domain $ R $ is equivalent to an invertible sheaf on its spectrum $ \mathrm{Spec}(R) $, so the Picard group $ \mathrm{Pic}(\mathrm{Spec}(F[W])) $ of invertible sheaves is precisely the ideal class group $ \mathrm{Cl}(F[W]) $ of invertible fractional ideals \cite[Example II.6.3.2]{hartshorne}. Note that an invertible $ R $-submodule of $ \mathrm{Frac}(R) $ is automatically a fractional ideal of $ R $ \cite[Theorem 11.6]{eisenbud}, so the ideal class group may also be defined purely in the language of invertible submodules.
\end{remark}

\pagebreak

\subsection{Construction of \texttt{to\_class}}

The function \texttt{to\_class} will map a nonsingular affine point $ (x, y) \in W(F) $ to the class of the invertible fractional ideal arising from the integral ideal $ \langle X - x, Y - y \rangle $. Defining the integral ideal explicitly is straightforward, and its associated fractional ideal is obtained by coercion.
\begin{lstlisting}
def XY_ideal (x : F) (y : F[X]) : ideal W.coordinate_ring :=
  ideal.span {adjoin_root.mk W.polynomial (C (X - C x)),
              adjoin_root.mk W.polynomial (Y - C y)}
\end{lstlisting}
Here, \texttt{ideal.span} constructs an integral ideal generated by the elements of a specified set, and \texttt{adjoin\_root.mk W.polynomial} is the canonical quotient map $ F[X, Y] \to F[W] $. Note also that \texttt{XY\_ideal} is defined slightly more generally than described, by allowing the second argument to be a polynomial in $ F[X] $ rather than just a constant.

On the other hand, checking that \texttt{XY\_ideal} is indeed invertible is slightly fiddly.

\begin{lemma}
\label{lem:XY_ideal_neg_mul}
For any $ (x, y) \in W(F) $,
\[ \langle X - x, Y - \sigma_x(y) \rangle \cdot \langle X - x, Y - y \rangle = \langle X - x \rangle. \]
\end{lemma}

\begin{proof}
Since $ W(x, y) = 0 $, there is an identity in $ F[W] $ given by
\[ (Y - y)(Y - \sigma_x(y)) \equiv (X - x)(X^2 + (x + a_2)X + (x^2 + a_2x + a_4) - a_1Y), \]
so the required equality may be reduced to $ \langle X - x \rangle \cdot I = \langle X - x \rangle $ in $ F[X, Y] $, where
\[ I := \langle X - x, Y - y, Y - \sigma_x(y), X^2 + (x + a_2)X + (x^2 + a_2x + a_4) - a_1Y \rangle. \]
Since $ (x, y) $ is nonsingular, either $ W_X(x, y) \ne 0 $ or $ W_Y(x, y) \ne 0 $, but
\[ W_X(x, y) = - (X + 2x + a_2)(X - x) + a_1(Y - y) + (X^2 + (x + a_2)X + (x^2 + a_2x + a_4) - a_1Y), \]
and $ W_Y(x, y) = -(Y - y) + (Y - \sigma_x(y)) $, so $ I = F[X, Y] $ in both cases.
\end{proof}

\begin{remark}
Geometrically, \cref{lem:XY_ideal_neg_mul} says that the line $ X = x $ intersects $ W $ at $ \mathcal{O}_W $ and at precisely two affine points $ (x, y) $ and $ (x, \sigma_x(y)) $, counted with multiplicity if they are equal.
\end{remark}

\cref{lem:XY_ideal_neg_mul} is \texttt{XY\_ideal\_neg\_mul}, and it follows that the fractional ideal $ \langle X - x, Y - y \rangle $ has inverse $ \langle X - x, Y - \sigma_x(y) \rangle \cdot \langle X - x \rangle^{-1} $. This is formalised as \texttt{XY\_ideal'\_mul\_inv}, which maps a proof that $ (x, y) \in W $ is nonsingular to a proof that the fractional ideal $ \langle X - x, Y - y \rangle $ has the specified right inverse. Passing this proof to \texttt{units.mk\_of\_mul\_eq\_one} then constructs the invertible fractional ideal of $ F[W] $ associated to $ \langle X - x, Y - y \rangle $.
\begin{lstlisting}
def XY_ideal' (h : W.nonsingular x y) :
  (fractional_ideal W.coordinate_ring⁰ W.function_field)ˣ :=
  units.mk_of_mul_eq_one _ _ (XY_ideal'_mul_inv h)
\end{lstlisting}

Now \texttt{to\_class} will be a \texttt{add\_monoid\_hom}, namely a function bundled with proofs that it maps zero to zero and respects addition. Its underlying unbundled function $ W(F) \to \mathrm{Cl}(F[W]) $, appropriately named \texttt{to\_class\_fun}, is defined separately to allow the equation compiler to generate lemmas automatically used in the proof that \texttt{to\_class} respects addition.
\begin{lstlisting}
def to_class_fun : W.point → additive (class_group W.coordinate_ring)
  | 0        := 0
  | (some h) := additive.of_mul (class_group.mk (XY_ideal' h))
\end{lstlisting}
Here, \texttt{additive G} creates a type synonym of a multiplicative group \texttt{G}, and the multiplicative \texttt{group} instance on \texttt{G} is turned into an additive \texttt{add\_group} instance on \texttt{additive G}. This is necessary to bundle \texttt{to\_class} as an \texttt{add\_monoid\_hom}, since \texttt{mathlib} does not have homomorphisms between an additive group and a multiplicative group by design.

\pagebreak

Now \texttt{to\_class\_fun} maps zero to zero by construction, but proving that it respects addition requires checking the five cases for \texttt{add} separately. The first two cases are trivial and the third case follows from \texttt{XY\_ideal\_neg\_mul}, while the last two cases are handled simultaneously by assuming \texttt{hxy} and checking an identity of integral ideals of $ F[W] $.

\begin{lemma}
\label{lem:XY_ideal_mul_XY_ideal}
For any $ (x_1, y_1), (x_2, y_2) \in W(F) $, if $ x_1 = x_2 $ implies $ y_1 \ne \sigma_{x_2}(y_2) $, then
\[ \langle X - x_1, Y - y_1 \rangle \cdot \langle X - x_2, Y - y_2 \rangle \cdot \langle X - x_3 \rangle = \langle X - x_3, Y - y_3 \rangle \cdot \langle Y - \lambda(X) \rangle, \]
where $ (x_3, y_3) := (x_1, y_1) + (x_2, y_2) $.
\end{lemma}

\begin{proof}
In both valid cases of \texttt{hxy}, the line $ Y = \lambda(X) $ contains $ (x_1, y_1) $ and $ (x_2, y_2) $, so
\[ \langle X - x_1, Y - y_1 \rangle = \langle X - x_1, Y - \lambda(X) \rangle, \qquad \langle X - x_2, Y - y_2 \rangle = \langle X - x_2, Y - \lambda(X) \rangle. \]
Furthermore, by (\ref{eq:1}) and the identity $ W(X, \lambda(X)) \equiv (Y - \lambda(X))(\sigma_X(Y) - \lambda(X)) $ in $ F[W] $, the required equality is reduced to checking that $ I := \langle X - x_1, X - x_2, Y - \sigma_X(Y) \rangle $ satisfies
\[ I \cdot \langle X - x_3 \rangle + \langle \sigma_X(Y) - \lambda(X) \rangle = \langle X - x_3, Y - y_3 \rangle, \]
where $ Y - \lambda(X) $ has been replaced by $ Y - \sigma_X(Y) $ in $ I $ since $ \sigma_X(Y) - \lambda(X) $ is present as a summand in the left hand side. By construction, the line $ Y = \lambda(X) $ contains $ (x_3, \lambda(x_3)) $, so the negated line $ \sigma_X(Y) = \lambda(X) $ contains its negation $ (x_3, \sigma_{x_3}(\lambda(x_3))) = (x_3, y_3) $. Then
\[ \langle X - x_3, Y - y_3 \rangle = \langle X - x_3, \sigma_X(Y) - \lambda(X) \rangle, \]
so it suffices to check that $ I = F[W] $. Now $ x_1 - x_2 = -(X - x_1) + (X - x_2) $, so $ I = F[W] $ if $ x_1 \ne x_2 $. Otherwise $ y_1 \ne \sigma_{x_1}(y_1) $, then there are no common solutions to $ Y = \sigma_{x_1}(Y) $ and $ W(x_1, Y) = 0 $, so $ I = F[W] $ by the Nullstellensatz. Explicitly, this follows from the identity
\[ (y_1 - \sigma_{x_1}(y_1))^2 \equiv -(4X^2 + (4x_1 + b_2)X + (4x_1^2 + b_2x_1 + 2b_4))(X - x_1) + (Y - \sigma_X(Y))^2 \]
in $ F[W] $, since $ W(x_1, y_1) = 0 $.
\end{proof}

\begin{remark}
Geometrically, the line $ Y = \lambda(X) $ intersects $ W $ at precisely three affine points $ (x_1, y_1) $, $ (x_2, y_2) $, and $ (x_3, \sigma_{x_3}(y_3)) $, which translates to the identity of integral ideals
\begin{equation}
\label{eq:2}
\langle X - x_1, Y - y_1 \rangle \cdot \langle X - x_2, Y - y_2 \rangle \cdot \langle X - x_3, Y - \sigma_{x_3}(y_3) \rangle = \langle Y - \lambda(X) \rangle.
\end{equation}
The identity in \cref{lem:XY_ideal_mul_XY_ideal} is then deduced by multiplying (\ref{eq:2}) with the identity in \cref{lem:XY_ideal_neg_mul} and cancelling $ \langle X - x_3, Y - \sigma_{x_3}(y_3) \rangle $ from both sides. Note that \cref{lem:XY_ideal_mul_XY_ideal} does not need the affine points to be nonsingular, while directly proving (\ref{eq:2}) does.
\end{remark}

\cref{lem:XY_ideal_mul_XY_ideal} is \texttt{XY\_ideal\_mul\_XY\_ideal}, and under these hypotheses, it follows immediately that the invertible fractional ideals $ \langle X - x_1, Y - y_1 \rangle $ and $ \langle X - x_2, Y - y_2 \rangle $ multiply to $ \langle X - x_3, Y - y_3 \rangle $ as classes in $ \mathrm{Cl}(F[W]) $, which along with \texttt{XY\_ideal\_neg\_mul} say that \texttt{to\_class} respects addition. The actual Lean proof is slightly technical, using the new library lemmas \texttt{class\_group.mk\_eq\_one\_of\_coe\_ideal} and \texttt{class\_group.mk\_eq\_mk\_of\_coe\_ideal} to reduce the equality between ideal classes arising from integral ideals to an equality between their underlying integral ideals up to multiplication by principal integral ideals, so the tactic mode proof below will only be sketched as a comment for the sake of brevity.
\begin{lstlisting}
@[simps] def to_class : W.point →+ additive (class_group W.coordinate_ring) :=
  { to_fun    := to_class_fun,
    map_zero' := rfl,
    map_add'  := /- Split the cases for P₁ + P₂. If P₁ = 0 or P₂ = 0, simplify.
                    Otherwise P₁ = (x₁, y₁) and P₂ = (x₂, y₂).
                      If x₁ = x₂ and y₁ = W.neg_Y x₂ y₂, use XY_ideal_neg_mul.
                      Otherwise use XY_ideal_mul_XY_ideal. -/ }
\end{lstlisting}

\pagebreak

\subsection{Injectivity of \texttt{to\_class}}

Injectivity is the statement that $ P_1 = P_2 $ if \texttt{to\_class} of $ P_1 $ equals \texttt{to\_class} of $ P_2 $ for any $ P_1, P_2 \in W(F) $, but a simple variant of \texttt{add\_left\_neg} shows that $ -P_1 + P_2 = 0 $ precisely when $ P_1 = P_2 $. Since \texttt{to\_class} is a \texttt{add\_monoid\_hom}, injectivity is equivalent to showing that \texttt{to\_class} of $ P $ is trivial implies $ P = 0 $ for any $ P \in W(F) $. In other words, it suffices to show that the integral ideal $ \langle X - x, Y - y \rangle $ is never principal for any affine point $ (x, y) \in W(F) $.

The approach taken circles around the fact that $ F[W] $ is a free $ F[X] $-algebra of finite rank, so it carries the notion of a \textbf{norm} $ \mathrm{Nm} : F[W] \to F[X] $. If $ f \in F[W] $, then $ \mathrm{Nm}(f) \in F[X] $ may be given by the determinant of left multiplication by $ f $ as an $ F[X] $-linear map, which is most easily computed by exhibiting an explicit basis $ \{1, Y\} $ of $ F[W] $ over $ F[X] $.
\begin{lstlisting}
lemma monic_polynomial : W.polynomial.monic
lemma nat_degree_polynomial : W.polynomial.nat_degree = 2

def basis : basis (fin 2) F[X] W.coordinate_ring :=
  (adjoin_root.power_basis' W.monic_polynomial).basis.reindex
    (fin_congr W.nat_degree_polynomial)
\end{lstlisting}
Here, \texttt{adjoin\_root.power\_basis'} returns the canonical basis of powers $ \{Y^i : 0 \le i < \deg_Y(W(X, Y))\} $, given the proof \texttt{monic\_polynomial} that $ W(X, Y) $ is monic. This is a type indexed by the finite type with $ \deg_Y(W(X, Y)) $ elements, which can be reindexed by the canonical finite type with two elements, since $ \deg_Y(W(X, Y)) = 2 $ by \texttt{nat\_degree\_polynomial}.

With this basis, any element $ f \in F[W] $ may be expressed uniquely as $ f = p(X) + q(X)Y $ with $ p(X), q(X) \in F[X] $, and the degree \footnote{\texttt{polynomial.degree} where $ \deg_X(0) := -\infty $ rather than \texttt{polynomial.nat\_degree} where $ \deg_X(0) := 0 $} of its norm can be computed directly.

\begin{lemma}
\label{lem:norm_smul_basis}
For any $ p(X), q(X) \in F[X] $,
\[ \deg_X(\mathrm{Nm}(p(X) + q(X)Y)) = \max(2\deg_X(p(X)), 2\deg_X(q(X)) + 3). \]
\end{lemma}

\begin{proof}
Let $ f := p(X) + q(X)Y $. In $ F[W] $ with the basis $ \{1, Y\} $ over $ F[X] $,
\begin{align*}
\mathrm{Nm}(f)
& \equiv \det\begin{pmatrix} p(X) & q(X) \\ q(X)(X^3 + a_2X^2 + a_4X + a_6) & p(X) - q(X)(a_1X + a_3) \end{pmatrix} \\
& = p(X)^2 - p(X)q(X)(a_1X + a_3) - q(X)^2(X^3 + a_2X^2 + a_4X + a_6).
\end{align*}
Let $ p := \deg_X(p(X)) $ and $ q := \deg_X(q(X)) $. Then
\[ \deg_X(p(X)^2) = 2p, \qquad \deg_X(q(X)^2(X^3 + a_2X^2 + a_4X + a_6)) = 2q + 3, \]
\[ \deg_X(p(X)q(X)(a_1X + a_3)) \le p + q + 1. \]
If $ p \le q + 1 $, then both $ p + q + 1 < 2q + 3 $ and $ 2p < 2q + 3 $, so $ \deg_X(\mathrm{Nm}(f)) = 2q + 3 $. Otherwise $ q + 1 < p $, then both $ p + q + 1 < 2p $ and $ 2q + 3 < 2p $, so $ \deg_X(\mathrm{Nm}(f)) = 2p $.
\end{proof}

\cref{lem:norm_smul_basis} is \texttt{norm\_smul\_basis}, and it follows by considering cases that $ \deg_X \mathrm{Nm}(f) \ne 1 $ for any $ f \in F[W] $, which is formalised as \texttt{nat\_degree\_norm\_ne\_one}.

\begin{remark}
Geometrically, $ \mathrm{Nm}(f) $ is the order of the pole of the rational function $ f \in F(W) $ at $ \mathcal{O}_W $. Using the norm allows for a purely algebraic argument for injectivity, which was inspired from an exercise in Hartshorne that assumes a short Weierstrass model where $ \mathrm{char}(F) \ne 2 $ \cite[Exercise I.6.2]{hartshorne}. This was also the last missing step in the whole argument, as JX only started computing degrees after he saw Borcherds's solutions to the exercise \cite{borcherds}.
\end{remark}

On the other hand, this degree is also the dimension of an $ F $-vector space.

\pagebreak

\begin{lemma}
\label{lem:finrank_quotient_span_eq_nat_degree_norm}
For any nonzero $ f \in F[W] $,
\[ \deg_X(\mathrm{Nm}(f)) = \dim_F(F[W] / \langle f \rangle). \]
\end{lemma}

\begin{proof}
In $ F[W] $ with the basis $ \{1, Y\} $ over $ F[X] $, multiplication by $ f $ as an $ F[X] $-linear map can be represented by a square matrix $ [f] $ over $ F[X] $, which has a Smith normal form $ M[f]N $, a diagonal matrix with diagonal entries some nonzero $ p(X), q(X) \in F[X] $, for some invertible matrices $ M $ and $ N $ over $ F[X] $. Now the quotient by $ f $ decomposes as a direct sum
\[ F[W] / \langle f \rangle \cong F[X] / \langle p(X) \rangle \oplus F[X] / \langle q(X) \rangle, \]
whose dimension as $ F $-vector spaces are precisely $ \deg_X(p(X)) $ and $ \deg_X(q(X)) $ respectively. On the other hand, the determinant of $ M[f]N $ is $ \det(M)\mathrm{Nm}(f)\det(N) = p(X)q(X) $, so
\[ \deg_X(\mathrm{Nm}(f)) = \deg_X(p(X)) + \deg_X(q(X)), \]
since the units $ \det(M), \det(N) \in F[X] $ are nonzero constant polynomials.
\end{proof}

\cref{lem:finrank_quotient_span_eq_nat_degree_norm} is \texttt{finrank\_quotient\_span\_eq\_nat\_degree\_norm}, and crucially uses the library lemma \texttt{ideal.quotient\_equiv\_pi\_span} to decompose the quotient by $ \langle f \rangle $ into a direct sum of quotients by its Smith coefficients. It is worth noting that the same argument clearly works more generally by replacing $ F[W] $ by any $ F[X] $-algebra with a finite basis. The proof of the injectivity of \texttt{to\_class} then proceeds by contradiction.

\begin{lemma}
\label{lem:to_class_injective}
The function $ W(F) \to \mathrm{Cl}(F[W]) $ is injective.
\end{lemma}

\begin{proof}
Let $ (x, y) \in W(F) $. It suffices to show that $ \langle X - x, Y - y \rangle $ is not principal, so suppose for a contradiction that it were generated by some $ f \in F[W] $. By \cref{lem:finrank_quotient_span_eq_nat_degree_norm},
\[ \deg_X(\mathrm{Nm}(f)) = \dim_F(F[W] / \langle f \rangle) = \dim_F(F[W] / \langle X - x, Y - y \rangle). \]
On the other hand, evaluating at $ (X, Y) = (x, y) $ is a surjective homomorphism $ F[X, Y] \to F $ with kernel $ \langle X - x, Y - y \rangle $, and this contains the element $ W(X, Y) $ since $ W(x, y) = 0 $. Explicitly, this follows from the identity in $ F[X, Y] $ given by
\[ W(X, Y) = (a_1y - (X^2 + (x + a_2)X + (x^2 + a_2x + a_4)))(X - x) + (y - \sigma_X(Y))(Y - y). \]
Thus by the first and third isomorphism theorems, there are $ F $-algebra isomorphisms
\[ F[W] / \langle X - x, Y - y \rangle \xrightarrow{\sim} F[X, Y] / \langle W(X, Y), X - x, Y - y \rangle = F[X, Y] / \langle X - x, Y - y \rangle \xrightarrow{\sim} F, \]
so $ \deg_X(\mathrm{Nm}(f)) = \dim_F(F) = 1 $, which contradicts \texttt{nat\_degree\_norm\_ne\_one}.
\end{proof}

\begin{remark}
\cref{lem:to_class_injective} can also be proven without the Smith normal form, by considering the ideal generated by the norms of elements in $ \langle X - x, Y - y \rangle $ for $ (x, y) \in W(F) $, namely
\[ I := \langle (X - x)^2, (X - x)((Y - y) + (\sigma_X(Y) - y)), (Y - y)(\sigma_X(Y) - y) \rangle. \]
On one hand, as an integral ideal in $ F[W] $, it can be shown that $ I $ is generated by the linear polynomial $ X - x $. On the other hand, if $ \langle X - x, Y - y \rangle $ were generated by some $ f \in F[W] $, then its ideal norm $ I $ is generated by $ \mathrm{Nm}(f) $, which cannot be linear by \cref{lem:norm_smul_basis}.
\end{remark}

\cref{lem:to_class_injective} is \texttt{to\_class\_injective}, and allows the proofs of commutativity and associativity in $ \mathrm{Cl}(F[W]) $ to be pulled back to $ W(F) $, thus proving \cref{prop:group}.
\begin{lstlisting}
lemma add_comm (P₁ P₂ : W.point) : P₁ + P₂ = P₂ + P₁
lemma add_assoc (P₁ P₂ P₃ : W.point) : (P₁ + P₂) + P₃ = P₁ + (P₂ + P₃)

instance : add_comm_group W.point :=
  ⟨zero, neg, add, zero_add, add_zero, add_left_neg, add_comm, add_assoc⟩
\end{lstlisting}

\pagebreak

\section{Discussion}
\label{sec:discuss}

\subsection{Related work}
\label{sec:relate}

As aforementioned, formalising the group law of an elliptic curve $ E $ over a field $ F $ is not novel, and has been done in several theorem provers to varying extents. Friedl (1998) \cite{friedl} gave a computational proof in the short Weierstrass model, leaving some of the heavy computations for associativity to CoCoA as a trusted oracle, and his argument was subsequently formalised by Th\'ery (2007) \cite{thery} in Coq. Fox, Gordon, and Hurd (2006) \cite{fox} formalised the addition law in the full Weierstrass model in HOL, but did not prove associativity. Hales and Raya (2020) \cite{hales} formalised a computational proof in Isabelle, but worked in the alternative Edwards model, which also fails to be an elliptic curve when $ \mathrm{char}(F) = 2 $.

The first known formalisation of an algebro-geometric proof was done by Bartzia and Strub (2014) \cite{bartzia}, who also worked in the short Weierstrass model. In 3,500 lines of Coq, they formalised the geometric notion of a Weil divisor \footnote{a formal $ \mathbb{Z} $-linear combination of points $ P \in E $ weighted by the order of vanishing of $ f $ at $ P $} of a rational function $ f \in F(E) $ to define the degree-zero Weil divisor class group $ \mathrm{Pic}^0(E) $, which is isomorphic to the Picard group $ \mathrm{Pic}(\mathrm{Spec}(F[E])) $ since $ E $ is nonsingular \cite[Corollary II.6.16]{hartshorne}. In another 6,500 lines of Coq, they constructed an analogous bijection between $ \mathrm{Pic}^0(E) $ and the points of $ E $ over the algebraic closure, but their argument is a simplification of the typical conceptual proof via the Riemann--Roch theorem and does not generalise easily to $ \mathrm{char}(F) = 2 $. In contrast, the algebraic proof with the ideal class group $ \mathrm{Cl}(F[E]) $ only spans 1,500 lines of Lean 3, avoiding the geometric theory and reusing much of the well-maintained algebraic libraries.

\subsection{Experimental attempts}
\label{sec:experiment}

The entire development process went through several iterations of trial and error, and various definitions of elliptic curves were proposed in Buzzard's topic on Zulip. The abstract definition as in \cref{rem} would be ideal, but algebraic geometry in \texttt{mathlib} is at its primitive stages, where describing properties of scheme morphisms like smoothness or properness, or defining the genus of a curve, would be a challenge. Since the Weierstrass model is universal over fields, the general consensus was that proving its equivalence with the abstract definition should proceed independently from proving theorems under the Weierstrass model.

Unfortunately, proving associativity became a huge issue in this model. The obvious first course of action is to check the equalities in all possible combinations of cases of addition, using the \texttt{field\_simp} and \texttt{ring} tactics to normalise rational expressions. In doing this, the number of cases quickly explode, and in the nontrivial cases of affine addition, the polynomial expressions involved become gargantuan. There are optimisations that could be made to reduce the number of cases, as coded by Masdeu \cite{masdeu} adapting Friedl's original argument into Lean, but a good way to manipulate the expressions remains elusive. In the generic case where three nonsingular affine points $ P_1, P_2, P_3 \in W(F) $ are in general position, \footnote{the affine points $ P_1 $, $ P_2 $, $ P_3 $, $ P_1 + P_2 $, $ P_1 + P_3 $, and $ P_2 + P_3 $ have pairwise distinct $ X $-coordinates} experiments by DKA with the aid of SageMath suggested that proving $ (P_1 + P_2) + P_3 = P_1 + (P_2 + P_3) $ by bashing out the algebra would involve polynomials each with tens of thousands of monomials, which is highly time-consuming in a formal system and definitely infeasible to work out by hand, despite taking only half a second in SageMath. The \texttt{ring} tactic, which uses proof generation by normalising to Horner form \cite{gregoire}, seems to be an order of magnitude too slow to work with such expressions, resulting in deterministic timeouts.

\pagebreak

The main culprits for the huge polynomials are the $ XY $ and $ Y $ terms in the Weierstrass equation, which do not allow even exponents of $ Y $ in the expressions to be substituted with polynomials solely in $ X $. When $ \mathrm{char}(F) \ne 2 $, these terms disappear with a change of variables, reducing the expressions to the computationally feasible range of hundreds of terms, hence enabling the work by Th\'ery (2007), or a transformation to the Edwards model whose group law was already formalised by Hales and Raya (2020). In principle, since $ 2 = 0 $ when $ \mathrm{char}(F) = 2 $, enough multiples of $ 2 $ may be cancelled from the expressions until a brute-force attack becomes feasible, but \texttt{mathlib} currently has no good tactic to do these cancellations except to manually extract these multiples of $ 2 $, such as by rewriting the expressions into the form $ p + 2q $ using \texttt{ring}, which is too slow in the first place, and deleting $ 2q $.

The mathematical literature typically deals with associativity by providing alternative proofs, in addition to the aforementioned algebro-geometric proof via the Picard group. One notable method goes via the uniformisation theorem in complex analysis \cite[Corollary VI.5.1.1]{silverman}, but \texttt{mathlib} also lacks much of the complex-analytic machinery to prove it, and this approach is only valid for $ \mathrm{char}(F) = 0 $ via the Lefschetz principle. Another approach uses the Cayley--Bacharach theorem in projective geometry \cite[Lemma 7.1]{cassels}, which proves associativity generically by a dimension counting argument. By continuing on Masdeu's branch, this approach seemed viable, but required redefining Weierstrass curves in projective coordinates and a convenient way to switch back to affine coordinates via dehomogenisation. Furthermore, the argument fails in a less generic case with a repeated point, which could be fixed by introducing an ad-hoc notion of intersection multiplicity between a line and a cubic, as suggested by Stoll. DKA started refactoring the definitions in an attempt at this approach, but ultimately switched to the current approach when proposed by JX. Note that an explicit exposition of a version of this argument can also be found in Nuida (2021) \cite{nuida}.

\subsection{Implementation issues}
\label{sec:implementation}

\subparagraph{Bivariate polynomials}

A bivariate polynomial in $ X $ and $ Y $ over a commutative ring $ R $ is typically represented in \texttt{mathlib} by a finitely supported function $ (\{0, 1\} \to \mathbb{N}) \to R $, associating a function $ f : \{0, 1\} \to \mathbb{N} $ to the coefficient of $ X^{f(0)}Y^{f(1)} $. This representation is very cumbersome when performing concrete manipulations, such as those in \cref{lem:XY_ideal_neg_mul} and \cref{lem:XY_ideal_mul_XY_ideal}, since explicit functions $ \{0, 1\} \to \mathbb{N} $ are needed to obtain coefficients.

In contrast, a polynomial in $ X $ over $ R $ is represented in \texttt{mathlib} by a finitely supported function $ \mathbb{N} \to R $, associating a natural number $ n \in \mathbb{N} $ to the coefficient of $ X^n $. A polynomial in $ Y $ with coefficients polynomials in $ X $ performs the same function as a bivariate polynomial in $ X $ and $ Y $, but the coefficient of $ X^nY^m $ is obtained by sequentially supplying two natural numbers $ m, n \in \mathbb{N} $. This has the additional advantage of aligning with the API for \texttt{adjoin\_root}, which gives a power basis needed in the proof of injectivity.

This representation does have the slightly awkward problem that $ X $ is denoted by \texttt{C X} while $ Y $ is denoted by \texttt{X}, but this is easily fixed by introducing \texttt{notation Y := X} and \texttt{notation R[X][Y] := polynomial (polynomial R)}. A more serious drawback is that existing results about multivariate polynomials, such as the Nullstellensatz, do not carry over to this representation, so explicit proofs with polynomial identities are sometimes necessary, namely in the proofs of \cref{lem:XY_ideal_neg_mul}, \cref{lem:XY_ideal_mul_XY_ideal}, and \cref{lem:to_class_injective}. Another issue is that the partial derivative with respect to $ X $ is obtained by applying the \texttt{polynomial.derivative} linear map to each coefficient of the polynomial in $ Y $, but the current \texttt{polynomial.map} only accepts a ring homomorphism, which explains why the partial derivatives \texttt{polynomial\_X} and \texttt{polynomial\_Y} were defined manually instead. In light of this, it has been suggested that \texttt{polynomial.map} should be refactored to accept set-theoretic functions instead.

\pagebreak

\subparagraph{Performance issues}

In the original definition of \texttt{to\_class}, it was observed that the function \texttt{class\_group.mk}, when applied to an invertible fractional ideal of \texttt{coordinate\_ring}, took a while to compile. Baanen diagnosed this problem and proposed the following solution \cite{angdinata}.
\begin{lstlisting}
local attribute [irreducible] coordinate_ring.comm_ring
\end{lstlisting}
Although \texttt{coordinate\_ring} is marked as \texttt{irreducible}, its \texttt{derive comm\_ring} tag generates a \texttt{reducible} instance of \texttt{comm\_ring}. In certain circumstances this is extremely slow, because the number of times an instance gets unified grows exponentially with its depth due to a lack of caching, and Baanen's solution was to force its \texttt{comm\_ring} instance to be irreducible locally whenever necessary. Note that this should have been fixed in Lean 4, and the port of \texttt{mathlib} to Lean 4 is expected to finish in a few months' time.

There are other performance issues that led to timeouts during development, but they were fixed by generalising the statements so they involve less complicated types.

\subparagraph{Proof automation}

The proofs of many basic lemmas often reduce to checking an equality of two polynomial expressions, such as in \cref{lem:nonsingular_neg} and \cref{lem:nonsingular_add}, but equality often holds only under some local hypotheses. Rather than rewriting these into the goal and applying the \texttt{ring} tactic, it is convenient to use \texttt{linear\_combination}, a newly-developed tactic that subtracts a linear combination of known equalities from the goal, before applying \texttt{ring}.

When several rewrite lemmas are often used together, it is also convenient to write a custom tactic to chain them. For instance, the evaluation map \texttt{eval} on a polynomial expression is often propagated inwards, so grouping the lemmas allows for a single tactic call.
\begin{lstlisting}
meta def eval_simp : tactic unit :=
`[simp only [eval_C, eval_X, eval_neg, eval_add, eval_sub, eval_mul, eval_pow]]
\end{lstlisting}

\subsection{Future projects}

Formalising the group law opens the doors to an expansive array of possible further work. An immediate project would be to enrich the API for nonsingular points by adding basic functorial properties with respect to a base change to a field extension $ K / F $. For instance, this could be defining the induced map $ E(F) \to E(K) $, or if $ K / F $ is Galois, computing the subgroup of $ E(K) $ invariant under the action of $ \mathrm{Gal}(K / F) $ to be precisely $ E(F) $.

It is worth noting the two ongoing projects by each of the two authors. DKA is formalising an inductive definition of division polynomials to understand the structure of the $ n $-torsion subgroup $ E[n] $ to compute the structure of the $ \ell $-adic Tate module $ \varprojlim_n E[\ell^n] $, while JX is formalising a proof that the reduction map $ E(K) \to E(R / \mathfrak{m}) $ is a group homomorphism for a discrete valuation ring $ R $ with fraction field $ K $ and maximal ideal $ \mathfrak{m} $.

In the longer run, one could explore the rich arithmetic theory over specific fields. Once the theory of local fields is sufficiently developed in \texttt{mathlib}, one could define the formal group of an elliptic curve, classify its reduction types, or state Tate's algorithm. These will be useful for the global theory, where one could define the Selmer and Tate--Shafarevich groups, give a Galois cohomological proof of the Mordell--Weil theorem, or state the full Birch and Swinnerton-Dyer conjecture. Over a finite field, one could verify the correctness of primality and factorisation algorithms as well as cryptographic protocols, or prove the Hasse--Weil bound or the Weil conjectures for elliptic curves.

Ultimately, a long term goal would be to redefine elliptic curves in \texttt{mathlib} as in \cref{rem} and prove \cref{prop:weierstrass}, but this will require a version of the Riemann--Roch theorem, whose proof will require a robust theory of sheaves of modules and their cohomology.

\pagebreak

\bibliography{main}

\end{document}